\documentclass[article]{amsart}
\date{\today}

\usepackage{amsmath,amsthm}
\usepackage{amssymb}
\usepackage{bbm}
\usepackage{mathrsfs}               
\usepackage{enumerate}               
\usepackage{color}

\newtheorem{corollary}{Corollary}
\newtheorem{proposition}{Proposition}
\newtheorem{lemma}{Lemma}

\newtheorem{theorem}{Theorem}

\newtheorem{remark}{Remark}

\newtheorem{example}{Example}

\newcommand{\N}{\mathbb{N}} 
\newcommand{\C}{\mathbb{C}} 
\newcommand{\Z}{\mathbb{Z}} 
\newcommand{\R}{\mathbb{R}} 

\newcommand{\bx}{{\bf x}}

\newcommand{\DA}{D_{{\bf A}}}
\newcommand{\SA}{  {\rm sgn}({\DA})}
\newcommand{\sps}[2]{\langle #1,#2 \rangle} 
\newcommand{\core}{C_0^\infty(\R^2,\C^2)}
\newcommand{\hilbert}{L^2(\R^2,\C^2)}
\newcommand{\lp}{\hat{\boldsymbol \pi}} 
\newcommand{\ri}{\mathrm{i\,}}

\newcommand{\rd}{\mathrm{d}}
\newcommand{\pn}{{p_n}}

\newcommand{\curla}{{\rm curl\,}{\bf A}}
\title[Confinement-deconfinement transitions for Dirac particles]{Confinement-deconfinement transitions for two-dimensional Dirac particles}
\author{Josef Mehringer}
\address{Josef Mehringer\\
Mathematisches Institut\\
Ludwig-Maximilians-Universit\"at\\
Theresienstra{\ss}e 39\\
D-80333 M\"unchen, Germany.}
\email{josef.mehringer@math.lmu.de}
\author{Edgardo Stockmeyer}
\address{Edgardo Stockmeyer\\
Mathematisches Institut\\
Ludwig-Maximilians-Universit\"at\\
Theresienstra{\ss}e 39\\
D-80333 M\"unchen, Germany.}
\address{
({\it Present address:}) Edgardo Stockmeyer\\ Departamento de F\'\i sica\\
Pontificia Universidad Cat\'olica de Chile\\
Vicu\~na Mackenna 4860\\
 Santiago 7820436, Chile.}
\email{stock@fis.puc.cl}

\begin{document}
\begin{abstract}
  We consider a two-dimensional massless Dirac operator coupled to a
  magnetic field $B$ and an electric potential $V$ growing at
  infinity. We find a characterization of the spectrum of the
  resulting operator $H$ in terms of the relation between $B$ and $V$
  at infinity. In particular, we give a sharp condition for the
  discreteness of the spectrum of $H$ beyond which we find dense pure
  point spectrum.
\end{abstract}
\maketitle
\section{Introduction}
Graphene, a two dimensional lattice of
carbon atoms arranged in a honeycomb structure, has attracted great attention in the last few years
due to its unusual properties \cite{castro2009electronic,novoselov2004electric}.  The dynamics of
its low-energy excitations (the charge carriers) can be described by a
two-dimensional massless  Dirac operator $D_0$ \cite{wallace,Fefferman}, where the speed of
light $c$ is replaced by the Fermi velocity, $v_F\sim 10^{-2} c$.  A
remarkable property of these Dirac particles is their lack of
localization in the presence of electric potential walls (i.e.,
potentials $V$ with $V(\bx)\to\infty$ as $|\bx|\to\infty$).  Indeed,
if we assume that $V$ is rotationally symmetric and of `regular
growth' the spectrum of the operator $D_0+V$ equals the
whole real line and is absolutely continuous
\cite{titchmarsh,newton,erdelyi,kms97}.  It is also known, at least in
three dimensions, that a much larger class of potentials growing at
infinity do not produce eigenvalues \cite{vogelsang,Kalfetal2003}.

One way to localize Dirac particles (in the sense that the Hamiltonian
has non-trivial 
discrete spectrum)  is through inhomogeneous magnetic
fields when they are asymptotically constant
\cite{de2007magnetic,KS2012} as well  as when they grow to infinity.  
Consider, for instance, a magnetic field $B=\curla$ with $B(\bx)\to\infty$ as $|\bx|\to\infty$ and
denote by $\DA$ the corresponding Dirac operator coupled to $B$.  It
is known that the spectrum of $\DA$ is discrete away from zero and
zero is an isolated  eigenvalue
of infinite multiplicity (see Section \ref{susy}).  

In this article we consider two-dimensional massless Dirac operators
coupled to both an electric potential $V$ and a magnetic field $B$. We
study the combination of the two effects described above: The
deconfinement effect associated to $V$ and the confinement one
associated to $B$.

Before presenting our main results let us first discuss this problem
assuming that $V$ and $B$ are sufficiently regular  positive rotationally
symmetric functions. In this case, the
Dirac operator admits an angular momentum decomposition $ \DA
+V=\bigoplus_{j\in \Z} h_j$ and its spectrum, $\sigma(\DA+V)$,
satisfies
\begin{align*}
  \sigma(\DA+V)=\overline{\bigcup_{j\in \Z} \sigma{(h_j)}}.
\end{align*}
Let 
$
  A(r)=\frac{1}{r}\int_0^r B(s)s ds
$
be the  modulus of  the magnetic vector potential in the rotational gauge.
It is easy to show, on the one hand, that if $A(r)\to \infty$ as $r\to\infty$ and 
\begin{align}\label{sym1}
  \lim_{|\bx|\to \infty} V(|\bx|)/A(|\bx|)<1,
\end{align}
the spectrum of  $h_j$ is discrete, for each $j\in \Z$  (see Proposition
\ref{purepoint} in the appendix for the precise statement).  On the other hand,
as opposed to the non relativistic case, it is  known
\cite[Proposition 2]{KMSYamada} that if $V(r)\to \infty$ as
$r\to\infty$ and 
\begin{align}\label{sym2}
   \lim_{|\bx|\to \infty} V(|\bx|)/A(|\bx|)>1,
\end{align}
the spectrum of each $h_j$ equals the whole real line and is purely
absolutely continuous.  This
phenomenon was recently discussed from the physical point of view in
\cite{giavaras2009magnetic}. In that article a device was proposed  to
control the localization properties of particles in graphene by
manipulating the electro-magnetic field at infinity, i.e., far away
from the sample.

Clearly, if condition \eqref{sym2} is satisfied the spectrum of the
full operator $H:=\DA+V$ is also absolutely continuous and equals
$\R$. Conversely, one has pure point spectrum if condition \eqref{sym1} holds.
It is, however,  unclear whether the eigenvalues of $H$ accumulate. 
Assume that $V(\bx), B(\bx)\to \infty$ as $|\bx|\to \infty$. 
One may expect that when the quotient $|V(|\bx|)/A(|\bx|)|$ is
sufficiently small, for large $|\bx|$,  the main effect of the
electric potential is to remove the zero modes of $\DA$ yielding
purely discrete spectrum for $H$. However, as this quotient
grows the eigenvalues of the $h_j$ might accumulate creating points in the essential
spectrum of $H$.

The aim of our work is to shed some light on the spectrum of
$H$ in terms of the relation between $B$ and $V$ at infinity.  We
emphasize that most of our results do not assume rotational
symmetry. In fact, besides some regularity conditions on $B$ and $V$
we only require that the potential $V$ grows subexponentially fast.

Let us  describe  our results  disregarding technical  assumptions (for the precise
 statements see Section \ref{mainresults}).   
We show (see Theorem \ref{thm1}) that the spectrum of $H$ is discrete if
\begin{align}
  \label{eq:33.2}
 \limsup_{|\bx|\to \infty} \frac{V^2(\bx)}{2|B(\bx)|}<1
\end{align}
and moreover that this condition is sharp in the following sense: 
If the quotient in \eqref{eq:33.2} converges to
$1$ along a sequence $(\bx_n)_{n\in \N}\subset \R^2$ with
$\lim_{n\to\infty}|\bx_n|= \infty$ the essential spectrum of $H$,
$\sigma_{\rm ess}(H)$, is not empty. In fact, we prove  (see
Theorem \ref{thm3a}) that zero belongs to
the essential spectrum of $H$ if, for some natural number $k$, the
quotient $V^2(\bx_n)/|2B(\bx_n)|$ 
converges to $k$ sufficiently fast, as $n\to \infty$. In addition,  we
find (see Corollary \ref{thm3-2}) that the essential spectrum of $H$ covers
the whole real line if there is a continuous path
$\gamma:\R^+\to \R^2$ with $|\gamma(t)|\to \infty$ as $t\to \infty$
such that  $ V^2(\gamma(t))/|2B(\gamma(t))|$
converges  to infinity  with moderate speed (see Remark \ref{growth}), as $t\to \infty$.

In order to get a better picture of our results consider the example
when  $V(\bx) = V_0|\bx|^t$ and  $B(\bx) =
B_0|\bx|^s$ for some constants $V_0,B_0, t>0$ and $s\ge 0$. In this case ${\bf A}$, in the rotational
gauge, satisfies 
$$|{\bf A}(|\bx|)|=B_0 |\bx|^{s+1}/(s+2).$$
Thus,  according to
Proposition \ref{purepoint} and \cite[Proposition
5.1]{KMSYamada}   (see also \eqref{sym1} and
\eqref{sym2}) we have that 
\begin{itemize}
\item[(a)] $\sigma(H)$ is pure point   if either $t< s+1$, or  $ t=s+1$
  and $V_0 < B_0/(s+2).$
\item[(b)] $\sigma(H)=\R$ is purely absolutely  continuous if either
  $t > s+1$, or  $ t=s+1$ and $V_0 > B_0/(s+2)$.
\end{itemize}
From Theorems \ref{thm1} and \ref{thm3a} and Corollary \ref{thm3-2} 
(see Example \ref{mainexample}) we conclude that
\begin{itemize}
  \item[(c)] $\sigma_{\rm ess}(H)=\emptyset$  if either $t < s/2$, or
    $ t=s/2$ and $V_0^2 < 2B_0$.
 \item[(d)] $0\in \sigma_{\rm ess}(H)$  if $t = s/2$ and $V_0^2 =
    2kB_0$ for some $k\in \N$.
\item[(e)] $\sigma(H)=\R$ is pure point if $t \in (s/2,(s+1)/2)$.
\end{itemize}
Hence, if we fix $B$ and increase the strength of $V$ (at infinity),
we observe a transition from purely discrete (c) to purely absolutely
continuous spectrum (b) passing through dense pure point spectrum (e).
In other words, we observe a transition between a strongly confined
system, in the sense that the energy required to `bring a particle to
infinity' is not finite, to a completely deconfined system. Between
these two regimes there is another one (e) belonging to a (presumably) weaker form of
confinement.\\
\vspace{0.5mm}\\
{\it The organization of this article is as follows:} In the next
section we state  our main results precisely. We recall some basic
facts about magnetic Dirac operators in Section \ref{susy}. In Section
\ref{usefull} we prove some useful commutator estimates used in the
proof of Theorem \ref{thm1} which is given in Section
\ref{prooft1}. Theorem \ref{thm2} is proven in Section \ref{prooft2}
and the proofs of Theorems \ref{thm3a} and \ref{thm3} are given in
Section \ref{prooft3}.  The main text is followed by an appendix containing
some auxiliary results.
\bigskip

\noindent
{\bf Acknowledgments.}
This work has been supported by SFB TR12 of the DFG.
\section{Main results}\label{mainresults}
Assume that $V$ and 
$B$ are continuous functions on $\R^2$.  We define the two-dimensional
massless Dirac operator coupled to a magnetic field $B$ on $\mathcal{H} :=
\hilbert$ a priori as
\begin{align}
  \label{eq:3}
  \DA\varphi:={\boldsymbol \sigma}\cdot \big(- \ri \nabla- {\bf
    A}\big)\varphi,\quad \varphi\in \core,
\end{align}
where ${\bf A}=(A_1,A_2)\in C^1(\R^2, \R^2)$ satisfies $B=\curla:=\partial_1A_2 -\partial_2
A_1$ and ${\boldsymbol \sigma}=(\sigma_1,\sigma_2)$ with 
\begin{align*}
  \sigma_1=\left(\begin{array}{cc}
0&1\\
1&0 
\end{array}\right),
\quad
 \sigma_2=\left(\begin{array}{cc}
0&-\ri\\
\ri&0 	
\end{array}\right).
\end{align*}
Similarly we define 
\begin{align}
  \label{eq:4}
  H\varphi:=(\DA+V)\varphi,\quad \varphi\in \core.
\end{align}
In view of \cite{Chernoff77} $\DA$ and  $H$ are
essentially self-adjoint. We denote the self-adjoint extensions
of $\DA$ and $H$ by the same symbols and their domains by
$\mathcal{D}(\DA)$
and $\mathcal{D}(H)$ respectively.
\begin{theorem}\label{thm1}
Let $V\in C^1(\R^2,\R)$, $B\in C(\R^2,\R)$, and ${\bf
  A}\in C^1(\R^2,\R^2)$ such that $B=\curla$ on $\R^2$. Assume that
\begin{eqnarray}
&&|V(\bx)|\to \infty\quad\mbox{as}\quad|\bx|\to \infty,\label{t1}\\
 &&\left|\frac{\nabla V (\bx)}{V(\bx)}\right|
  \to 0\quad\mbox{as}\quad|\bx| \to \infty,\label{t2}\\
&&\limsup_{|\bx|\to \infty}
  \frac{V^2(\bx)}{2|B(\bx)|}<1. \label{t3}
\end{eqnarray}
Then $(H+\ri)^{-1}$ is  compact  in $\mathcal{H}$, i.e., 
  $\sigma_{\rm ess}(H)=\emptyset$.
\end{theorem}
\begin{remark}\label{suzuki}
  A similar result was obtained in \cite{Suzuki2000}. However, there
  the statement was only proved when the limit superior in \eqref{t3}
  is strictly smaller than $1/4$ instead of $1$. Our proof is quite
  different from the one given in \cite{Suzuki2000}. We split the 
  analysis on the spaces ${\rm ker}\DA$ and ${\rm ker}\DA^\perp$
and estimate the cross terms with the commutator bounds derived in
Section \ref{usefull}. 
\end{remark}
The  next two theorems state  that the constant  $1$ in \eqref{t3} above is in
fact sharp. 
\begin{theorem}\label{thm2}
Let $ V\in C^1(\R^2,\R)$,  $ B\in C^2(\R^2,\R)$, with $|B|>B_0$ for some $B_0>0$ and  ${\bf
  A}\in C^3(\R^2,\R^2)$ such that $B=\curla$ on $\R^2$. Assume that
\begin{eqnarray}
&& |V(\bx)|\to \infty\quad\mbox{as}\quad|\bx|\to \infty,
\label{wi1}\\
&& \|\nabla V/V\|_\infty<\infty\quad\mbox{and}\quad \|\Delta B/B\|_\infty<\infty, \label{wi2}\\
&& \|(V^2 -2 |B|)/V\|_\infty<\infty.\label{wi3} 
\end{eqnarray}
Then $\sigma_{\rm ess}(H)\not= \emptyset.$
\end{theorem}
\begin{remark}\label{velocity1}
Note that conditions \eqref{wi1} and \eqref{wi3} are equivalent to 
\begin{align*}
  \left| \frac{V^2(\bx)}{2|B(\bx)|}-1\right| \le  \, \frac{c}{|V(\bx)|} \longrightarrow 0 \quad 
\mbox{as}\quad |\bx|\to\infty,
\end{align*}
for some constant $c > 0$.
\end{remark}
\begin{remark}
The statement of  Theorem \ref{thm2} can also be
obtained assuming that $\liminf_{|\bx|\to \infty} |\bx|^2\big(B(\bx) - \tfrac{1}{2}\Delta B(\bx)/ B(\bx)\big) >0 $
instead of\eqref{wi1} and only that $\|\nabla
V/V\|_\infty<\infty$ instead of \eqref{wi2}.
\end{remark}
In order to state the next two results we use the following definition:
We say that a function $f:\R^2\to\R$ varies with 
rate $\nu\in [0,1]$ at infinity if there are constants $R>1$ and $C>0$ such that for any
$\alpha:\R^2\to\R^2$ with $\alpha(\bx)=o(|\bx|^\nu)$ as
$|\bx|\to \infty$ the function $f$ satisfies the bound
\begin{align}
  \label{eq:12}
  |f(\bx +\alpha(\bx))|\le C |f(\bx)|, \quad \mbox{for all} \quad |\bx|>R.
\end{align}
Clearly, if $f$ varies with rate $\nu\in [0,1]$ then it also varies with
rate $\nu'$ at infinity  for all $\nu '\in [0,\nu]$.
Note also  that power functions of $|\bx|$ with positive power
vary with any rate $\nu\in [0,1]$.
\begin{theorem}\label{thm3a}
  Let $B,V\in C^1(\R^2,\C)$ such that $|\nabla V|, |\nabla B|$ vary with rate $0$
  at infinity. Let ${\bf A}\in
  C^2(\R^2,\R^2)$ such that $B=\curla$ on
  $\R^2$.  Assume that there is a sequence
  $(\bx_n)_{n\in\N}$ with $|\bx_n|\to\infty$ as $n\to\infty$ and constants $k\in \N$ and $\varepsilon\in(0,1)$ such that, as
$n\to\infty$, 
\begin{eqnarray}
\label{con0a}
 &&|V(\bx_n)|\to \infty,\\
\label{con1a}
&& \frac{\nabla B(\bx_n)}{|B(\bx_n)|^{1-\varepsilon}},
\,\frac{\nabla V(\bx_n)}{|V(\bx_n)|^{1-\varepsilon}}\longrightarrow 0,
\\
\label{con2a}
 &&\frac{V^2(\bx_n)-2k|B(\bx_n)|}{V(\bx_n)}\longrightarrow
0.
\end{eqnarray}
Then $0\in \sigma_{\rm ess}(H)$.
\end{theorem}
\begin{remark}\label{velocity}
 From \eqref{con0a} and \eqref{con2a} it follows that $|B(\bx_n)|\to \infty$ as well as 
\begin{align*}
  \left| \frac{V^2(\bx_n)}{2|B(\bx_n)|}-k\right| \le  \, \frac{c}{|V(\bx_n)|} \longrightarrow 0 \quad 
\mbox{as}\quad |\bx|\to\infty,
\end{align*}
for some constant $c > 0$. Note, however, that the converse is not in
general true.
\end{remark}
We also remark that, under somewhat different assumptions, Theorem
\ref{thm3a} (with $k=1$) improves the statement of Theorem
\ref{thm2}. The proof of Theorem \ref{thm2} is based on the construction
of an infinite dimensional subspace of the operator domain on which
$H$ stays bounded. This space is constructed using the
zero-modes of the operator $\DA$. The proof of Theorem \ref{thm3a} is,
in contrast, based on the construction of a Weyl sequence of
functions  localized around points $(\bx_n)_{n\in\N}$ where
$|V(\bx_n)| \approx \sqrt{2k|B(\bx_n)|}$, i.e., the points where the
potential $V$ has the same value as the $k-th$ Landau-level of the
magnetic Dirac operator with constant field $B(\bx_n)$. This idea of
$V$ crossing through the (local) Landau-levels can also be used in the
case when $V^2(\bx)/2B(\bx) \to \infty$ as $|\bx| \to \infty$ (along
some sequence) to obtain the following result.
\begin{theorem}\label{thm3}
  Let $B,V\in C^1(\R^2,\C)$ such that $|\nabla V|$ and $|\nabla B|$ vary with rate $\nu\in[0,1]$ at
  infinity. Let ${\bf A}\in
  C^1(\R^2,\R^2)$ such that $B=\curla$ on
  $\R^2$. Assume that there is a sequence
  $(\bx_n)_{n\in\N}$ with $|\bx_n|\to\infty$ as $n\to\infty$ and
  constants $\varepsilon,\alpha,\kappa, B_0>0$ such that, as $n\to\infty$,
\begin{eqnarray}\label{con1}
&& \frac{V^2(\bx_n)-2n|B(\bx_n)|}{V(\bx_n)}\longrightarrow
0,\\
\label{con2}
&&\frac{\nabla B(\bx_n)}{B(\bx_n)}\left(
  \frac{V^2(\bx_n)}{2|B(\bx_n)|}\right)^{1+\varepsilon}, \frac{\nabla V(\bx_n)}{V(\bx_n)}\left(
  \frac{V^2(\bx_n)}{2|B(\bx_n)|}\right)^{1+\varepsilon}
\longrightarrow 0,
\\
\label{con24}
&&
\frac{1}{|\bx_n|^{2 \nu} |B(\bx_n)|} \left(
  \frac{V^2(\bx_n)}{2|B(\bx_n)|}\right)^{1+\varepsilon}\longrightarrow 0,
\end{eqnarray}
and furthermore, for all $n\in \N$, 
  \begin{eqnarray}
   \label{con0}
&& B_0\le|B(\bx_n)|\le \alpha\left(\frac{V^2(\bx_n)}{2|B(\bx_n)|}\right)^{\kappa}.
\end{eqnarray}
Then $0\in \sigma_{\rm ess}(H)$. 
\end{theorem}
\begin{remark}\label{growth}
  Due to \eqref{con1} and \eqref{con0} we have that $|V(\bx_n)|\to \infty$
  and moreover  
\begin{align*}
  \left| \frac{V^2(\bx_n)}{2n|B(\bx_n)|}-1\right| \le  \, \frac{c}{|V(\bx_n)|} \longrightarrow 0 \quad 
\mbox{as}\quad |\bx|\to\infty,
\end{align*} for some constant $c>0$.
Observe  that conditions \eqref{con2} and \eqref{con24} give an upper bound for the
growth of  the ratio $V^2(\bx_n)/|2B(\bx_n)|$.

Note in addition that the theorem is also applicable for  bounded
magnetic field.  In this case the condition
  \eqref{con24} can only be satisfied for $\nu>0$.
\end{remark}
\begin{remark}
  It is easy to see that the regularity conditions on $V$ and $B$ in
  theorems \ref{thm3a} and \ref{thm3}  can be weakened to hold
  only outside some compact set $K\subset \R^2$. Inside $K$ it is
  sufficient that these functions are bounded. The same holds true for
  Theorem \ref{thm1} (compare with Lemma \ref{schneider}).  
\end{remark}
The theorem above can also be used to find other points in the
essential spectrum of $H$. To this end is suffices to find a sequence
$(x_n)_{n\in\N}$ satisfying the conditions of the theorem  for $V-E$ instead of $V$. As
an example we get the following result.
\begin{corollary}\label{thm3-2}
  Let $B,V\in C^1(\R^2,\R)$ such that $|\nabla V|$ and $|\nabla B|$
  vary with rate $s\in [0,1]$ at
  infinity. Let ${\bf A}\in
  C^1(\R^2,\R^2)$ such that $B=\curla$ on
  $\R^2$.  Assume that there is a continuous path
  $\gamma:\R^+\to\R^2$ with $|\gamma(t)|\to\infty,$ as $ t\to\infty,$ 
such that 
\begin{equation}
    \label{eq:27}
    \frac{V^2(\gamma(t))}{2B(\gamma(t))} \longrightarrow \infty \quad
    \mbox{as} \quad t\to\infty.
  \end{equation}
  Moreover, assume that the conditions \eqref{con2}, \eqref{con24}, and
  \eqref{con0} 
are satisfied for any sequence $(\bx_n)_{n\in\N}$ in
  the range of $\gamma$ with $|\bx_n|\to \infty$.  Then $\sigma(H)=\R$. 
\end{corollary}
For completeness we give a proof of this corollary.
\begin{proof}
  Let $E\in \R$. Due to \eqref{eq:27} and the continuity of $V^2/B$
  along $\gamma$ we find a constant $N>0$ and a sequence
  $(\bx_n)_{n\in\N}$, on the range of $\gamma$ with $|\bx_n|\to \infty$,
  such that
\begin{equation}
  \label{eq:23}
  (V-E)^2(\bx_n)=2nB(\bx_n),\quad\mbox{for all} \quad n >N.
\end{equation}
Since $|V(\bx_n)|\to\infty$ as $n\to\infty$ it is clear that the 
conditions \eqref{con2}, \eqref{con24}, and
  \eqref{con0}  are also satisfied for $V(\bx_n)-E$ instead of
  $V(\bx_n)$. This implies the claim.
\end{proof}
\begin{remark}
  As we already mentioned in the Introduction one can combine
  Corollary \ref{thm3-2} and Proposition \ref{purepoint} from the
  Appendix for functions $V$ and $B$ that are rotationally
  symmetric. In this case one obtains that $\sigma(H)=\R$ is pure
  point, i.e., $H$ has dense pure point spectrum (see Example
  \ref{mainexample} bellow). Note that the same type of spectral
  phenomenon occurs for $\sigma(\DA)$ when $B$ is rotationally
  symmetric and decays at infinity but $B(\bx)|\bx|\to \infty$ as
  $|\bx|\to \infty$ \cite{miller1980quantum} (see also
  \cite[pp. 208]{Thaller}).
\end{remark}
Let us now apply Corollary \ref{thm3-2} to some particular cases.
\begin{example}\label{mainexample}
  Let $V(\bx) = V_0|\bx|^t$ and  $B(\bx) = B_0|\bx|^s$ 
  for some constants $V_0,B_0>0$ and $0\le s <2t$. Then $|\nabla V|$ and $|\nabla B|$ vary with
  any rate $\nu\in [0,1]$ at infinity  and we have that $V^2(\bx)/(2B(\bx))\to\infty$ as
  $|\bx|\to \infty$. Condition \eqref{con2} is satisfied for any
  sequence if and only if 
  $$\frac{1}{|\bx|}\left(
  \frac{V^2(\bx)}{2|B(\bx)|}\right)^{1+\varepsilon}\longrightarrow
0\quad  \mbox{as} \quad|\bx| \to \infty,$$
 which is the case whenever
  $2t < s +1$. For these exponents  \eqref{con24} and \eqref{con0} are clearly
  fulfilled for any sequence which tends to infinity.
 Hence, Corollary \ref{thm3-2} states that for $0\le s < 2t< s + 1$
  we have $\sigma(H)=\R$. Furthermore, in view of Proposition
  \ref{purepoint} we get that the spectrum in this case is pure point.
\end{example} 
\section{Supersymmetry and zero modes}\label{susy}
In this section we recall some basic facts about magnetic Dirac
operator which are going to be useful in our proofs.  As we mentioned
in the Section \ref{mainresults} the operator $\DA$ defined in \eqref{eq:3} is
essentially self-adjoint on $\core$.  We write
\begin{align}
  \label{eq:5}
  \DA=:\left(
\begin{array}{cc}
0&d^*\\
d&0  
  \end{array}
  \right),
\end{align}
where 
$$d:=\overline{p_1-A_1
+\ri(p_2-A_2)\upharpoonright}_{C^\infty_0(\R^2,\C)} $$ 
 and $p_j=-\ri \partial_j, \,j=1,2$, is   the momentum operator in the
$j$-th direction. The operators $d, d^*$ satisfy the commutation relation 
\begin{align}
  \label{eq:1}
  [d,d^*]\varphi:=(dd^*-d^*d)\varphi=2B\varphi\,, \quad \varphi\in\mathcal{D}(d^*d)
\cap \mathcal{D}(dd^*).
\end{align}
We now investigate further the relation between $dd^*$ and $d^* d$.  We
 note that the of kernel  $\DA$ fulfills ${\rm ker} (\DA)= {\rm
  ker}(d)\oplus {\rm ker}(d^*)$. Due to
the matrix structure of $\DA$ we have  that
\begin{align}
  \label{eq:8}
  {\rm sgn}({\DA}):=\frac{\DA}{|\DA|}=\left(\begin{array}{cc}
0&s^*\\
s&0  
\end{array}\right)
  \quad \mbox{on}\quad {\rm ker} (\DA)^\perp,
\end{align}
with the  unitary operators 
\begin{align}\label{super0}
  &s:{\rm ker}(d)^\perp\to{\rm ker}(d^*)^\perp, \qquad s^*: {\rm ker}(d^*)^\perp\to {\rm ker}(d)^\perp
\end{align}
(see \cite[Section 5.2.3]{Thaller} for a related discussion). Using the operator identity $\DA^2=
\SA\DA^2\SA$ we find, for any $\varphi=(\varphi_1,\varphi_2)^{\rm T}$
with $\varphi_1\in \mathcal{D}({d^*d})\cap {\rm ker}(d)^\perp$ and
$\varphi_2\in \mathcal{D}({dd^*})\cap {\rm ker}(d^*)^\perp$,
that 
\begin{align}
\label{super}
\left(
 \begin{array}{cc}
d^*d&0\\
0&dd^*  
  \end{array}
  \right)\varphi = \left(\begin{array}{cc}
s^*dd^*s&0\\
0&sd^*ds^*  
  \end{array}
  \right) \varphi.
\end{align}

We now recall some results concerning the structure of the kernel of
$\DA$. For a given  $B\in C(\R^2,\R)$  one finds a weak solution $\phi\in
C^1(\R^2,\R) $ of the equation
$$\Delta \phi = B,$$ 
see, e.g. \cite{Erdoes2002} where a much larger class of fields $B$ is
considered. (Using standard elliptic regularity it is easy to see that
if $B$ belongs to some H\"older class the solution $\phi\in C^2(\R^2,\R)$).
A direct computation yields
\begin{align}
  \label{eq:7}
  {\rm  ker} (d^*d)=\{\omega e^{-\phi}\in L^2(\R^2,\C)\,|\,\omega\,\, \mbox{is
    entire in}\,\, x_1+\ri x_2\}.
\end{align}
Moreover, 
it is known  \cite{Rozenblum2006}  that whenever 
\begin{align}
  \label{eq:6}
  \int_{\R^2} [B(\bx)]_+ d\bx=\infty,\quad  \int_{\R^2} [B(\bx)]_- d\bx<\infty,
\end{align}
zero is an eigenvalue of infinite multiplicity of $d^*d$, i.e., 
 ${\rm dim\, ker} (d^*d)=\infty$.

With these observations we can easily go through the following example
that will be useful later on.
\begin{lemma}\label{lemma1}
  Let $B\in C(\R^2,\R)$ such that $B\ge B_0>0$ and 
${\bf A}\in C^1(\R^2,\R^2)$ with
  $B=\curla$. Then  $0$ is an
  isolated  eigenvalue of
infinite multiplicity of $\DA$. In addition, 
\begin{align}\label{zero}
  {\rm  ker}(\DA)=\Bigg\{ \left(
\begin{array}{c}
\Omega\\
0
\end{array}
\right)\in \hilbert\,|\, \Omega\in {\rm  ker} (d^*d)\Bigg\}\,.
\end{align}
Moreover, 
\begin{align*}
  (-\sqrt{2B_0},0)\cup(0,\sqrt{2B_0})\subset \varrho(\DA), 
\end{align*}
where $\varrho(\DA)$ denotes the resolvent set of $\DA$. 
\end{lemma}
\begin{proof}
Due to
\eqref{eq:1} we have the operator inequality 
\begin{align}
  \label{eq:10}
  dd^* \ge 2B\ge 2B_0.
\end{align}
Hence, $ {\rm  ker} (d^*)= {\rm  ker} (dd^*)=\{0\}$ and we find
\begin{align*}
   {\rm  ker} (\DA)=  {\rm  ker} (d)\oplus   {\rm  ker} (d^*)
={\rm  ker} (d^*d)\oplus  \{0\}.
\end{align*}
This yields  \eqref{zero}.
Moreover, using again \eqref{eq:10}, the isospectrality
\begin{align*}
  \sigma(d^*d)\setminus\{0\}= \sigma(dd^*)\setminus\{0\},
\end{align*}
together with
\begin{align*}
  \DA^2\varphi=\left(\begin{array}{cc}
d^*d & 0\\
0& d d^*
\end{array}
\right)\varphi,\quad \varphi\in \core,
\end{align*}
imply  that $(0,2B_0)\subset \varrho(\DA^2)$. Therefore,  by
the spectral theorem, we find the desired spectral gap.
\end{proof}
\begin{remark}
If we assume further that $B(\bx)\to \infty$ as $|\bx|\to \infty$ we
have  by   \cite{HelfferNourrigatWang1989} that the spectrum of $\DA$
is discrete away from zero. Since $B$ fulfills in this case \eqref{eq:6} this implies
that $\sigma_{\rm ess}(\DA)=\{0\}$.
\end{remark}
\section{Useful commutator estimates}\label{usefull}
We denote by $P_0$ the orthogonal projection onto ${\rm ker} (\DA)$ and
set $P_0^\perp:=1-P_0$.  In this section we show some commutator
bounds between the electric potential $V$, $P_0$, and the sign of
$\DA$ denoted by ${\rm sgn}(\DA)$. We use these bounds in Section
\ref{prooft1} to show Theorem \ref{thm1}.

Throughout this section we use the following notation:
For $0<V\in C^1(\R^2,\R)$ such that $\|\nabla V/V\|_\infty<\infty$ we
set
\begin{equation*}
T := \frac{\ri{\boldsymbol\sigma}\cdot \nabla V}{V}.
\end{equation*}
Note that $T$  formally equals  $[\DA, V^{-1}]V$, where
 $[\cdot, \cdot]$ is the symbol for the commutator.
We define 
for $z\in \varrho(\DA)$ 
\begin{align}
  \label{eq:9}
  R_{\bf A}(z):=(\DA-z)^{-1}.
\end{align}
\begin{lemma}\label{resolvent1}
Let $0<V\in C^1(\R^2,\R)$, $B\in C(\R^2,\R)$, and 
${\bf A}\in C^1(\R^2,\R^2)$ with
  $B=\curla$. Assume further that $\|V^{-1}\|_\infty, \|\nabla
  V/V\|_\infty<\infty $ and that  $z\in
  \varrho(\DA)$ is such that $\| T R_{\bf A}(z)\|<1$. Then
  \begin{align}
    \label{eq:11}
   \left[ R_{\bf A}(z),V^{-1} \right] &=- V^{-1}  R_{\bf A}(z) T R_{\bf A}(z) \Theta_r(z)\\
\label{eq:11.1}&=\Theta_l(z) R_{\bf A}(z) T  R_{\bf A}(z) V^{-1}, 
  \end{align}
where 
\begin{align*}
  \Theta_r(z)&:= \big( 1+ T R_{\bf A}(z) \big)^{-1},\\
\Theta_l(z)&:= \big( 1-  R_{\bf A}(z) T  \big)^{-1},
\end{align*}
and 
\begin{align}
  \label{eq:11.2} 
\|\Theta_r(z)\|, \|\Theta_l(z)\|\le \big( 1-\|T R_{\bf A}(z)\| \big)^{-1}.
\end{align}
\end{lemma}
\begin{proof}
 For $z\in \varrho(\DA)$ we get the following
  relation on $\mathcal{H}$
\begin{align*}
  \left[ R_{\bf A}(z),V^{-1} \right] &= R_{\bf A}(z)\left[ V^{-1},\DA \right] R_{\bf A}(z)\\
&= - R_{\bf A}(z) V^{-1} T R_{\bf A}(z)\\
&= - V^{-1}  R_{\bf A}(z) T R_{\bf A}(z) - \left[ R_{\bf A}(z), V^{-1} \right] T R_{\bf A}(z).
\end{align*}
From this follows that 
\begin{align*}
   \left[ R_{\bf A}(z),V^{-1} \right] \big(1+ T R_{\bf A}(z) \big)=
   -  V^{-1}  R_{\bf A}(z) T R_{\bf A}(z) .
\end{align*}
Since $\|T R_{\bf A}(z))\|<1$ we get that $\pm1\in \varrho ( T R_{\bf A}(z))$ by the
  Neumann series. Therefore,  we get
  the desired expression \eqref{eq:11} and the estimate on
  $\Theta_r(z)$. Equation \eqref{eq:11.1} and the bound for
  $\Theta_l(z)$  follows similarly. 
\end{proof}
As stated in Lemma \ref{schneider} in the appendix we can reduce the
proof of 
Theorem  \ref{thm1} for potentials $V\in C^1(\R^2,\R)$ and the magnetic fields $B\in
C(\R^2,\R)$ satisfying the following conditions: There exist constants
$\delta,\eta\in (0,1)$ such that, for all $\bx\in\R$,
\begin{align}
\label{c1} &V(\bx)\ge 1/\delta,\\
\label{c2} &|\nabla V(\bx)|\le \delta V(\bx),\\
\label{c3} &V^2(\bx)\le 2(1-\eta)B(\bx).
\end{align} 
Due to Lemma \ref{lemma1} we see that under these
assumptions $0$ is an isolated eigenvalue of $\DA$ of infinite
multiplicity and that $\sigma(\DA)\setminus\{0\}$ is
discrete. Moreover, we find a spectral gap  $(-2\beta_0,0)\cup(0,2\beta_0)\subset
\varrho(\DA)$, where
\begin{align}
  \label{eq:14}
\beta_0:= \big(2\delta\sqrt{1-\eta} \ \big)^{-1}.  
\end{align}
\begin{lemma}\label{lemmakey}
Let $V\in C^1(\R^2,\R)$, $B\in C(\R^2,\R)$, and 
${\bf A}\in C^1(\R^2,\R^2)$ with
  $B=\curla$. Assume further that the
  conditions \eqref{c1}-\eqref{c3} are fulfilled for
  $\delta\in(0,\frac{1}{2})$ and $\eta\in(0,1)$. Then, we have
  \begin{itemize}
  \item[(a)] The operators $\left[ P_0^\perp,V^{-1} \right] V$ and
    $V\left[ P_0^\perp,V^{-1}\right]$ are well-defined on  $\core$ and  extend to
    bounded operators on $\mathcal H$ with 
    \begin{align}
      \label{eq:13}
      \left \lVert V\left[P_0^\perp,V^{-1}\right]\right \rVert , \left \lVert \left[P_0^\perp,V^{-1}\right] V\right \rVert
       \le 4 \delta^2.    
    \end{align}
The same holds true if we replace $P_0^\perp$ above by $P_0$. \\
\item[(b)] $P_0\mathcal{D}(V), P_0^\perp \mathcal{D}(V) \subset
  \mathcal{D}(V) $. \\
\item[(c)] The operator $V\left[P_0^\perp,V^{-1}\right]$ maps $\core$ in
  $\mathcal{D}(\DA)$.  Moreover, we have 
  \begin{equation}
    \label{eq:18}
      \left \lVert \DA V\left[P_0^\perp,V^{-1}\right] \right \rVert \le  4 \delta.
  \end{equation} 
  \end{itemize}
\end{lemma}
\begin{proof}
{\it Part (a):} Let $\varphi,\psi\in \core$ with
$\|\varphi\|=\|\psi\|=1$. Using the representation formula for
the spectral projection and  that $(-2\beta_0,0)\cup(0,2\beta_0)\subset
\varrho(\DA)$ we have 
\begin{align*}
  P_0=-\frac{1}{2\pi\ri} \int_{|z|=\beta_0} R_{\bf A}(z)  \rd z. 
\end{align*}
Using  the estimate $\|R_{\bf A}(z)\|\le 1/\beta_0$, condition \eqref{c2}, and
that $\delta\in (0,\frac{1}{2})$, we get  for $|z|=\beta_0$, 
\begin{equation*}
\|T R_{\bf A}(z))\| \le \|T \| \|R_{\bf A}(z)\| \le \delta \beta_0^{-1}
= 2\delta^2 \sqrt{(1-\eta)} \le 1/{2}.
\end{equation*}
Thus, applying   Lemma \ref{resolvent1} we get 
\begin{align*}
  |\sps{V\varphi}{\left[ P_0^\perp,V^{-1}\right] \psi}|&=
  |\sps{V\varphi}{\left[ P_0,V^{-1} \right] \psi}|\\
  &\le \frac{1}{2\pi}\int_{|z|=\beta_0} |\sps{V\varphi}{\left[ R_{\bf A}(z),V^{-1} \right] \psi}| \ \mbox{d}z\\
&=\frac{1}{2\pi}\int_{|z|=\beta_0} |\sps{\varphi}{R_{\bf A}(z) T R_{\bf A}(z) \Theta_r(z)\psi}| \ \mbox{d}z\\
  &\le \frac{1}{2\pi} \int_{|z|=\beta_0} \frac{\|T\| 
\|R_{\bf A}(z)\|^2}{1-\|T\| \|R_{\bf A}(z)\|} \mbox{d}z \ \le 4 \delta^2 .
\end{align*}
The estimate for
$\left[ P_0^\perp,V^{-1}\right] V$ follows analogously.

{\it Part (b):} From the previous computation we see that $\left[ P_0^\perp,V^{-1} \right]$ maps $\mathcal H$ on
$\mathcal{D}(V)$. The claim is therefore an  immediate  consequence of the identity 
\begin{align*}
  P_0\varphi=V^{-1}P_0V\varphi+ \left[ P_0,V^{-1} \right] \varphi,\quad \varphi\in \mathcal{D}(V).
\end{align*}

{\it Part (c):}  By the spectral theorem we get, for $|z|=\beta_0$, 
\begin{align*}
  \|\DA R_{\bf A}(z) \|=\sup_{\lambda\in
    \sigma(\DA)}\Big|\frac{\lambda}{\lambda-z}\Big|\le 2.
\end{align*}
Proceeding similarly as in the proof of part (a) we have, for  $\varphi,\psi\in \core$ with $\|\varphi\|=\|\psi\|=1$,
\begin{align*}
   |\sps{\DA\varphi}{V[P_0^\perp,V^{-1}]\psi}|&\le \frac{1}{2\pi}\int_{|z|=\beta_0} |\sps{\varphi}{\DA R_{\bf A}(z) 
                          TR_{\bf A}(z) \Theta_r(z)\psi}| \ \mbox{d}z 
  \le  \ 4\delta.
\end{align*}
\end{proof}
\begin{lemma}\label{lemmakey2} 
  Let $V\in C^1(\R^2,\R)$, $B\in C(\R^2,\R)$, and ${\bf A}\in
  C^1(\R^2,\R^2)$ with $B=\curla$. Assume
  further that the conditions \eqref{c1}-\eqref{c3} are fulfilled for
  $\delta\in(0,\frac{1}{2})$ and $\eta\in (0,1)$. Then $\left[{\rm
      sgn}(\DA)P_0^\perp,V^{-1} \right]$ maps $\hilbert$ in
  $\mathcal{D}(V)$ and 
\begin{equation}\label{rrr}
\left \lVert V\left[{\rm sgn}(\DA)P_0^\perp,V^{-1} \right] \right
\rVert \leq  4 \delta^2.
\end{equation}
\end{lemma}
\begin{proof}
Let  $P_I $ denote the spectral projection of $\DA$ on 
the interval $I \subset \R $. Then, we 
have
\begin{equation*}
 {\rm sgn}(\DA) P_0^\perp  = \ P_{(\beta_0,\infty)} \ - \ P_{(-\infty,-\beta_0)}
\end{equation*}
with $\beta_0$ as in \eqref{eq:14}.
Using  \cite[Lemma VI-5.6]{Kato1980}  we find the representations
\begin{equation*}
 P_{(-\infty, -\beta_0)} =  \tfrac{1}{2} \big( \mathbbm{1}  -  U(-\beta_0) \big), \
\end{equation*}
\begin{equation*}
 P_{(\beta_0, \infty)} =  \tfrac{1}{2} \big( \mathbbm{1}  +
 U(\beta_0) \big),   \end{equation*}
where, for $\lambda\in\varrho(\DA)$,  
\begin{align*}
  U(\lambda)={\rm s}-\lim_{R\to \infty} \int_{-R}^{R} R_{\mathbf A} (\lambda+\ri
  t) \frac{\rd t}{\pi}=:\int_{-\infty}^\infty R_{\mathbf A} (\lambda+\ri
  t) \frac{\rd t}{\pi}.
\end{align*}
This yields  the commutator identity
\begin{equation*}
  \left[ {\rm sgn}(\DA)P_0^\perp,V^{-1} \right]   = \ \tfrac{1}{2} \left[ U(-\beta_0), V^{-1} \right]
  + \tfrac{1}{2} \left[ U(\beta_0), V^{-1} \right] . 
\end{equation*}
Pick $\varphi,\psi \in \core$ with $\|\varphi\|=\|\psi\|=1$. 
 Applying Lemma \ref{resolvent1} we get
\begin{align*}
  |\sps{V\varphi}{\left[U(\beta_0),V^{-1} \right]\psi}| &=
\frac{1}{\pi} \left| \int_{-\infty}^{\infty}
 \sps{V\varphi}{\left[R_{\bf A}(\beta_0 + \ri t), V^{-1} \right] \psi } \ \mbox{d}t \right|  \\ & \leq
\frac{1}{\pi}  \int_{-\infty}^{\infty}
 \left \lVert \varphi \right \rVert  \left \lVert V\left[R_{\bf A}(\beta_0 + \ri t), V^{-1} \right] \psi \right \rVert  
\mbox{d}t
\\ & = 
\frac{1}{\pi}  \int_{-\infty}^{\infty} \left \lVert R_{\bf A}(\beta_0 + \ri t) T R_{\bf A}(\beta_0 + \ri t)
  \Theta_r(\beta_0 + \ri t) \psi \right \rVert  \mbox{d}t
\\ & \leq
\frac{\delta}{\pi}  \int_{-\infty}^{\infty} \left \lVert R_{\bf A}(\beta_0 + \ri t) \right \rVert^2 
\left \lVert \Theta_r(\beta_0 + \ri t) \right \rVert
  \mbox{d}t
\\ & \leq
\frac{2\delta}{\pi}  \int_{-\infty}^{\infty} \frac{1}{\beta_0^2 + t^2} \ \mbox{d}t  =\frac{2\delta}{\beta_0} \le  4\delta^2.
\end{align*}
In the estimate above we use that, for $t \in \R$,
\begin{equation*}
 \left \lVert \Theta_r(\beta_0 + \ri t) \right \rVert \leq 
\big( 1 - \left \lVert R_{\bf A}(\beta_0 + \ri t) \right \rVert  \lVert T \rVert  \big)^{-1} \leq
 \big( 1 - \delta \left \lVert R_{\bf A}(\beta_0 ) \right \rVert \big)^{-1} \leq 2.
\end{equation*}
Similarly, the same inequality holds for $\left[U(-\beta_0),V^{-1} \right]$.
Hence, we find that 
\begin{equation*}
 |\sps{V\varphi}{\left[ P_0^\perp {\rm sgn}(\DA) P_0^\perp,V^{-1} \right]\psi}| \leq 4 
\delta^2, 
\end{equation*}
From this follows the
proof of the lemma.
\end{proof}
\section{Proof of Theorem \ref{thm1}}\label{prooft1}
We note that the assumptions in Theorem 1 imply either that $V(\bx)\to
\infty$ as $|\bx|\to \infty $ or $V(\bx)\to - \infty$ as $|\bx|\to
\infty $ by using the continuity of $V$.  We may assume without lost
of generality
that $V$ is positive at infinity.  Similarly, it suffices to consider
the case $B(\bx)\to \infty$ as $|\bx|\to \infty$ since otherwise we
just have to change
the roles of $d$ and $d^*$ in the proof. 

In order to prove Theorem \ref{thm1} it suffices to find a constant $c>0$ such that 
\begin{align}
  \label{eq:2}
  \|(\DA+V) \varphi \|\ge c\|V\varphi\|,\quad \varphi\in \core,
\end{align}
holds; see, e.g. \cite{Suzuki2000}. Moreover, according to Lemma \ref{schneider} we
may assume  that  $V$ and $B$ fulfill the conditions
\eqref{c1}-\eqref{c3}, where  $\eta\in(0,1)$ is some fix constant and 
 $\delta\in(0,1)$ can be chosen arbitrarily small.

 Since $\sigma_{\rm ess}(\DA)=\{0\}$ it is convenient to show
 \eqref{eq:2} by splitting $\varphi$ as the sum of $P_0\varphi$ and
 $P_0^\perp\varphi$. Using the bounds derived in Section \ref{usefull}
 we can then estimate the cross terms.
\begin{proof}[Proof of Theorem \ref{thm1}]
Let $\varphi\in\core$. We compute, using Lemma \ref{lemmakey},
\begin{equation*}
\begin{split}
  \left \lVert (\DA+V)\varphi \right \rVert ^2 &= \left \lVert (\DA+V)(P_0+P_0^\perp)\varphi \right \rVert^2\\
&= \left \lVert \big(VP_0+(\DA+V)P_0^\perp \big)\varphi \right \rVert^2\\
&= \left \lVert(\DA+V)P_0^\perp\varphi \right \rVert^2+2{\rm Re}\sps{(\DA+V)P_0^\perp\varphi}{VP_0\varphi}+
\left \lVert VP_0\varphi \right \rVert^2\\
&=\|(\DA+V)P_0^\perp\varphi\|^2-\delta \left \lVert VP_0^\perp\varphi \right \rVert^2\\
&\quad+
2{\rm Re}\sps{VP_0\varphi}{\DA
  P_0^\perp\varphi}+\|V\varphi\|^2-(1-\delta) \left \lVert VP_0^\perp\varphi \right \rVert^2.
\end{split}
\end{equation*}
We estimate each of the terms above separately. Observe that for
any $\varepsilon\in(0,1)$ we have that  
\begin{equation}
  \label{eq:16}
\begin{split}
  &\|(\DA+V)P_0^\perp\varphi\|^2-\delta \left \lVert VP_0^\perp\varphi \right \rVert^2\\
  &\qquad\qquad\ge (1-\varepsilon)\|\DA
  P_0^\perp\varphi\|^2+(1-\varepsilon^{-1}-\delta)
\|V P_0^\perp\varphi\|^2.
\end{split}
\end{equation}
An application of Lemma \ref{lemmakey}  yields
\begin{equation}
  \label{eq:17}
  \begin{split}
|\sps{VP_0\varphi}{\DA
  P_0^\perp\varphi}|&=|\sps{VP_0V^{-1} V\varphi}{\DA
  P_0^\perp\varphi}|\\
&= |\sps{P_0 V\varphi}{\DA
  P_0^\perp\varphi}+\sps{V \left[P_0,V^{-1} \right] V\varphi}{\DA P_0^\perp\varphi}|\\
&=|\sps{V\left[ P_0,V^{-1} \right] V\varphi}{\DA P_0^\perp\varphi}|\\
&\le 4\delta\, \|V\varphi\|\,\|\varphi\|  \\
&\le 4\delta^2\, \|V\varphi\|^2,
  \end{split}
\end{equation}
where in the last equality we use \eqref{c1}.
Further, by Lemma \ref{lemmakey} (a), we obtain
\begin{align*}
  \left \lVert VP_0^\perp\varphi \right \rVert \le  \left \lVert V\varphi \right \rVert
  + \left \lVert V \left[P_0^\perp,V^{-1} \right] V\varphi \right \rVert 
 \le (1+ 4\delta^2)\|V\varphi\|.
\end{align*}
Thus
\begin{align}
  \label{eq:19}
  \|V\varphi\|^2-(1-\delta)\left \lVert VP_0^\perp\varphi \right \rVert^2 \ge (\delta - 4\delta^2) \|V\varphi\|^2.
\end{align}
Therefore, in view of \eqref{eq:16},\eqref{eq:17}, and \eqref{eq:19}
it suffices to show  that 
\begin{align}
  \label{eq:20}
  \left\| \DA P_0^\perp\varphi \right\|^2
    +\tfrac{(1-\varepsilon^{-1}-\delta)}{(1-\epsilon)}
    \left\| VP_0^\perp\varphi \right\|^2  \ge 0,
\end{align} 
for $\delta>0$ small enough and some $\varepsilon \in (0,1)$. We set
\begin{align*}
  c_{\epsilon,\delta}:=-\tfrac{(1-\varepsilon^{-1}-\delta)}{(1-\epsilon)}>0.
\end{align*}
In view of Lemma \ref{lemma1} we have that 
\begin{align}
  \label{eq:21}
  P_0^\perp=\left(
\begin{array}{cc}
\lp&0\\
0&\mathbbm{1}
\end{array}\right),
\end{align}
where $\lp$ denotes the orthogonal  projection onto ${\rm
  ker}(d)^{\perp}$. Using this, \eqref{eq:5}, and writing
$\varphi=(\varphi_1,\varphi_2)^{\rm T}$ we get 
\begin{align*}
  \left \lVert \DA P_0^\perp\varphi \right \rVert^2- c_{\epsilon,\delta}\left \lVert V P_0^\perp\varphi \right \rVert^2=
\|d\lp\varphi_1\|^2- c_{\epsilon,\delta}\|V\lp\varphi_1\|^2+\|d^*\varphi_2\|^2- c_{\epsilon,\delta}\|V\varphi_2\|^2.
\end{align*}
According to condition \eqref{c3} and
\eqref{eq:1} we have that
\begin{equation}\label{ddstar}
\begin{split}
  \|d^*\varphi_2\|^2-
  c_{\epsilon,\delta}\|V\varphi_2\|^2&=\sps{\varphi_2}{dd^*\varphi_2} -
  c_{\epsilon,\delta}\sps{\varphi_2}{V^2\varphi_2}\\
&\ge [1-c_{\varepsilon,\delta}(1-\eta)] \sps{\varphi_2}{2B\varphi_2}.
\end{split}
\end{equation}
In order to give a lower bound to $ \|d\lp\varphi_1\|^2-
c_{\epsilon,\delta}\|V\lp\varphi_1\|^2$ we will use that 
$$dd^*=sd^* d
s^*\quad \mbox{on} \quad {\rm Ran}(\lp)\cap \mathcal{D}(dd^*),$$ where $s,s^*$ are the isometries
given in  \eqref{eq:8}. A simple computation yields 
\begin{equation}\label{dddd}
\begin{split}
  \|d\lp\varphi_1\|^2- c_{\epsilon,\delta}\|V\lp\varphi_1\|^2&=
  \|ds^*s\lp\varphi_1\|^2- c_{\epsilon,\delta}\|Vs^*s\lp\varphi_1\|^2\\
&= \|d^*s\lp\varphi_1\|^2- c_{\epsilon,\delta}\|Vs^*V^{-1}
Vs\lp\varphi_1\|^2\\
&=\|d^*s\lp\varphi_1\|^2- c_{\epsilon,\delta}\big\|\big(s^*+V\left[s^*,V^{-1}\right]\big)
Vs\lp\varphi_1\big\|^2.
\end{split}
\end{equation}
We note that 
$V\left[s^*,V^{-1}\right]$
is one of the components of the operator
\begin{align*}
  V \left[{\rm sgn}(\DA)P_0^\perp,V^{-1}\right]= \left(\begin{array}{cc}
0&V\left[s^*,V^{-1}\right]\\
V\left[s\lp, V^{-1}\right]&0
  \end{array}
\right).
\end{align*}
Using the definition of the operator norm, we obtain that
\begin{align*}
\left \lVert V \left[s^*,V^{-1}\right] \right \rVert
&= \sup_{\varphi=(0,\varphi_2)^{\rm T}, \|\varphi_2\|=1 } 
 \left \lVert V \left[{\rm sgn}(\DA)P_0^\perp,V^{-1}\right] \varphi \right \rVert
\\
&\le
    \left \lVert V \left[{\rm sgn}(\DA)P_0^\perp,V^{-1}\right] \right
    \rVert\le 4\delta^2,
\end{align*}
where in the last bound we use Lemma \ref{lemmakey2}. 
Combining this with
\eqref{dddd} and proceeding as in \eqref{ddstar} we obtain 
\begin{align}\label{sstar}
  \|d\lp\varphi_1\|^2- c_{\epsilon,\delta}\|V\lp\varphi_1\|^2&
\ge \left[1- c_{\epsilon,\delta}(1-\eta)(1+ 12\delta^2)\right]
\sps{s\lp\varphi_1}{2B s\lp\varphi_1}.
\end{align}
Choosing $\varepsilon=1-\delta^{1/2}$ we get that
$c_{\epsilon,\delta}=1+\mathcal{O}(\delta^{1/2})$ as $\delta\to
0$. This implies that, for $\delta>0$ sufficiently small, the terms
in \eqref{ddstar} and \eqref{sstar} are positive. This concludes the
proof of the theorem.
\end{proof} 
\section{Proof of Theorem \ref{thm2}}\label{prooft2}
In this section we prove Theorem \ref{thm2} for $B>0$. The case $B<0$
can be done similarly. 
 As we mention in Section \ref{susy} there is a function 
$\phi\in C^2(\R^2, \R)$
satisfying $\Delta \phi=B$. We choose
${\mathbf A}(\bx)=(-\partial_2\phi, \partial_1 \phi)$ as the vector
potential in the Hamiltonian $\DA$. A key element in our proof is
that the space 
\begin{equation}\label{paula}
X :=\{ \omega e^{- \phi} \in L^2(\R^2,\C; B dx)\,|
\,\omega\,\, \mbox{is entire in}\,\, x_1+\ri x_2\}\, , 
\end{equation}
is infinite dimensional. Since $B>B_0$ we see that $X$ is a subspace of ${\ker(d)}$  (see \eqref{eq:7}).
Let us first state a technical result
concerning this space whose proof is given at the end of this section.
\begin{lemma}\label{domainquestions}
Assume that the conditions of Theorem \ref{thm2} are fulfilled. Let
$\phi \in C^2(\R^2, \R) $ be such that $B = \Delta \phi  $.
Then we have,  for   $\Omega \in X$,     
\begin{itemize}
\item[a)] $\Omega \in \mathcal{D}(d^*)\cap \mathcal{D}(V)$ and $ \left\lVert  d^* \Omega
  \right \rVert  = 
  \|\sqrt{2B}\Omega \|$
\item[b)] $(d^*\Omega, - V\Omega)^T  \in \mathcal{D}(H)$ 
\end{itemize}
\end{lemma}
 \begin{proof}[Proof of Theorem \ref{thm2}]
We first show that the dimension of $X$ (defined in \eqref{paula}) is
infinite.  To this end define
$\widetilde{\phi}=\phi-\frac{1}{2}\ln(B)$ and note that 
$$\Delta \widetilde{\phi}=B+\frac{(\nabla B)^2}{2B^2}-\frac{\Delta B}{2B}>B
-\frac{\Delta B}{2B}.$$
Thus, by the discussion in Section \ref{susy} (see \eqref{eq:6}) the space
 \begin{equation*}
Y:=\{ \omega e^{- \widetilde{\phi}} \in L^2(\R^2,\C)\,|
\,\omega\,\, \mbox{is entire in}\,\, x_1+\ri x_2\}\, , 
\end{equation*}
is infinite dimensional. The claim now follows since clearly $X$ and
$Y$ are isomorphic.
Let us define  the subspace
$$ W := \Bigg\{ \left( \begin{array}{c}\hspace{0.17cm} d^* \Omega \\   -V \Omega
                                   \end{array}
                                  \right)\in \hilbert\,|\, \Omega\in X  \Bigg\} \subset   \mathcal{D}(H). $$
By Lemma \ref{domainquestions} we have, for any  $\Omega \in X$, 
$$\left\lVert \left( \begin{array}{c}\hspace{0.17cm} d^* \Omega \\   -V \Omega
                              \end{array}  \right) \right\rVert  =  
         \left\lVert  d^* \Omega \right \rVert + \left\lVert  V \Omega
         \right \rVert 
 \ge  \sqrt{2B_0} \|\Omega \|.$$
Therefore, the map $X\ni \Omega\mapsto (d^*\Omega,-V\Omega)^{\rm T}\in
W$ is bijective and we conclude that $\dim W=\infty.$

Now pick $\psi = (d^*\Omega, - V\Omega)^T \in W$, then
\begin{align*}
\| H\psi \| &=  \left\lVert \left( \begin{array}{cc}  V&d^*\\  d&V  \end{array} \right)
                                         \left( \begin{array}{c}\hspace{0.17cm} d^* \Omega \\   -V \Omega \end{array}  
                         \right) \right\rVert   
                    = \left\lVert
\left( \begin{array}{c} 
(\ri\partial_1V+\partial_2V) \Omega \\   
(2B-V^2) \Omega\end{array}  
                         \right) \right\rVert \\ 
 & \le (\|\nabla V/ V\|_\infty+\|(2B-V^2)/V\|_\infty) \|\psi\| .
   \end{align*}
   Since $\dim W = \infty$ the inequality above implies that
   $\sigma_{ess} (H) \neq \emptyset$.
\end{proof}
\begin{proof}[Proof of Lemma \ref{domainquestions}]
Note first that from Remark \ref{velocity1} follows that  $X\subset\mathcal{D}(V)$. 
Let $\Omega\in X$ and $\chi \in C_0^\infty(\R^2,[0,1])$ be such that 
$$\chi(x) = 
\begin{cases}
  1,  & \text{for } |x| \leq 1\\
  0, & \text{for } |x| \geq 2\,.
\end{cases}$$ 
We define $\Omega_n := \chi(\frac{\cdot}{n})\Omega, n \in \N
$. Clearly, $\Omega_n \in C_0^2(\R^2,\C)$ holds and moreover $\Omega_n \to \Omega$ as $n \to \infty$ in $L^2(\R^2,\C
; B\mbox{d}x)$. Since $\Omega \in \mbox{ker}(d)$ we have,  using the
commutator identity \eqref{eq:1}, that
$$\|d^*\Omega_n \|^2  \ = \ \|d\Omega_n \|^2   + \  2\|\sqrt{B}\Omega_n \|^2  
\ = \ \frac{1}{n^2} \left \lVert (i\partial_1 \chi
  - \partial_2\chi)\left(\frac{\cdot}{n} \right) \Omega\right \rVert^2
+ \ 2\|\sqrt{B}\Omega_n \|^2 $$ holds,  for any $n \geq 1$.  Thus, 
$\big(d^*\Omega_n\big)_{n \in \N}$ is a Cauchy sequence in
$L^2(\R^2,\C )$. By
the closedness of $d^*$ we conclude that $\Omega \in \mathcal{D}(d^*)$
with $ \left\lVert  d^* \Omega \right \rVert  =  \|\sqrt{2B}\Omega \| $. 

For $\Omega \in X$ define $\psi := (d^*\Omega, - V\Omega)^T$.  A
direct computation, using integration by parts, shows that one finds  $\phi\in \hilbert$ such that for all
$\varphi\in \core$ holds 
\begin{align*}
\sps{H\varphi}{\psi}=\sps{\varphi}{\phi}.
\end{align*}
This implies the claim by the definition of the adjoint operator.
\end{proof}

\section{Proof of Theorems \ref{thm3a} and  \ref{thm3}}\label{prooft3}
We roughly explain the main ingredients of the proof of Theorems \ref{thm3a}
and \ref{thm3}: Note first that it suffices to construct a sequence
$(\psi_n)_{n\in\N} \subset \mathcal{D}(H)$ which converges weakly to
zero such that $\|H\psi_n\|/\|\psi_n\|\to 0$ as $n\to\infty$ (Weyl
sequence).  For simplicity we consider  the case when   $B(\bx)=B_0$. Recall that
the corresponding magnetic Dirac operator $\DA$ has the (Landau)
eigenvalues $l_n:=\mbox{sgn} (n) \sqrt{2|n|B_0}$, $n \in \Z$. Now
assume that there is a sequence $(\bx_n)_{n\in \N}$ such that
$V(\bx_n)=\sqrt{2nB_0}=l_n $ for all $n\in \N$ (this is a
simplification of condition \eqref{con1}). The bulk of the  Weyl sequence consists
of eigenfunction $(\varphi_n)_{n\in\N}$ of $\DA$ centered around $\bx_n$
and with eigenenergies $-l_n$. It is well known that these
eigenfunctions are almost Gaussian-like localized. Due to this strong
localization one has that $(V\varphi_n)(\bx)\approx V(\bx_n)\varphi_n(\bx) =l_n
\varphi_n(\bx)$ in the sense of $L^2$. Therefore, we get that
$$\|(\DA + V)\varphi_n\| \approx 0,$$
where the error terms will be controlled using the remaining
assumptions of the theorems.

This consideration can be also  extended to non-constant magnetic
fields. In this case we construct  the Weyl sequence using
eigenfunctions of the magnetic Dirac operator with constant magnetic
field $B(\bx_n)$. Such a local linearization can be well
controlled using again that the Landau  eigenfunction are strongly localized.


  Throughout this section we will assume without lost of generality  that $B(\bx_n)\ge B_0$
is positive. In order to prepare the proof of Theorems \ref{thm3a}
and \ref{thm3} we introduce some notation:  For a sequence 
$(\bx_n)_{n\in\N}\subset\R^2$ we set $V_n:=V(\bx_n)$ and $B_n:=B(x_n)$ and define the
magnetic vector potentials 
\begin{align*}
  &{\mathbf A}_n(\bx):=\int_0^1 B_n\wedge (\bx-\bx_n) s \rd
  s=\tfrac{1}{2} B_n\wedge (\bx-\bx_n) ,\\
  &\widetilde{{\mathbf A}}_n(\bx):= \int_0^1
  B(\bx_n+s(\bx-\bx_n))\wedge (\bx-\bx_n) s \rd s,
\end{align*}
where $a\wedge {\mathbf v}:=a(-v_2,v_1)$ , for $a\in \R$ and ${\mathbf v}=(v_1,v_2)\in \R^2$. Clearly,
\begin{align*}
  {\rm curl}\,{\mathbf A}_n=B_n\quad\mbox{and}\quad  {\rm
    curl}\,\widetilde{\mathbf A}_n=B\quad\mbox{on}\,\,\R^2.
\end{align*}
We define the operators $d_n$ and $d_n^*$ through the relation
\begin{align*}
 D_{{\mathbf A}_n}= \left(\begin{array}{cc}
0 & d_n^*\\
d_n & 0
    \end{array}
  \right).
\end{align*}
Let  $(\pn)_{n\in\N}$ be a
sequence of natural numbers. As already mentioned above an important
ingredient  in our proof is
that $2\pn B_n$ is the square of the $\pn$-th Landau level of the
Dirac operator $D_{{\mathbf A}_n}$.  
For all $n\in\N$, we define the functions 
\begin{align}
  \label{eq:15}
  \varphi_n(\bx):=\left(\begin{array}{c}
(\ri B_n (x_1-\ri x_2) )^{\pn} e^{-B_n|\bx|^2/4}\\
-V_n(\ri B_n (x_1-\ri x_2) )^{\pn-1} e^{-B_n|\bx|^2/4}
 \end{array}
  \right)
\end{align}
and observe that
\begin{align}
  \label{eq:24}
  \hat{\varphi}_n(\bx):= \varphi_n(\bx-\bx_n)=\left(\begin{array}{c}
(d_n^* )^\pn e^{-B_n|\bx-\bx_n|^2/4}\\
-V_n(d_n^*)^{\pn-1} e^{-B_n|\bx-\bx_n|^2/4}
 \end{array}
  \right).
\end{align}
We have, for any $k\in\N$ (see e.g., \cite[Section 7.1.3]{Thaller}), 
\begin{align}
  \label{eq:25}
  d_n d_n^*  [(d_n^* )^{k-1} e^{-B_n|\bx-\bx_n|^2/4}]=2kB_n [(d_n^* )^{k-1} e^{-B_n|\bx-\bx_n|^2/4}].
\end{align}
Next we define the localization functions.  Let $\chi \in
C^\infty_0(\R^2,[0,1])$ be such that $\chi(\bx)=1$ for $|\bx|\le1$
and $\chi(\bx)=0$ for $|\bx|\ge 2$. We set
\begin{align*}
  \chi_n(\bx):=\chi\left(\frac{\bx-\bx_n}{r_n}\right),
  \end{align*} 
  where  $r_n>0$ will be chosen in the proofs later on.

Finally, we observe that
  since ${\rm curl}\,( {\mathbf A}-\widetilde{{\mathbf A}}_n)=0$ there
  exists a function $g_n\in C^2(\R^2,\R)$ such that $\nabla
  g_n={\mathbf A}-\widetilde{{\mathbf A}}_n$ on $\R^2$.

We define the Weyl functions to be given, for all $n\in\N$, by
\begin{align}
  \label{eq:26}
  \psi_n(\bx):= e^{\ri g_n (\bx)}\chi_n(\bx) \hat{\varphi}_n(\bx),
  \quad \bx\in\R^2.
\end{align}
Clearly, we have 
\begin{equation}\label{weyl}
\begin{split}
 e^{-\ri g_n} (\DA+V)\psi_n&=(D_{\widetilde{\mathbf
     A}_n}+V) \chi_n \hat{\varphi}_n\\
&=\chi_n(D_{{\mathbf
     A}_n}+V_n)\hat{\varphi}_n-\ri ({\boldsymbol \sigma}\cdot\nabla
\chi_n)\hat{\varphi}_n\\
&\quad+{\boldsymbol \sigma}\cdot({\mathbf
  A}_n-\widetilde{\mathbf A}_n) \chi_n \hat{\varphi}_n+
(V-V_n)\chi_n\hat{\varphi}_n.
\end{split}
\end{equation}
In order to prove  Theorems \ref{thm3a} and \ref{thm3} we will choose
$r_n>0$ such that $\psi_n$ are linear independent and 
each of the four terms above tend to zero uniformly in $\|\psi_n\|$ as
$n\to\infty$. In the next lemma we estimate the norms
$\|\psi_n\|$. 
\begin{lemma}\label{landau}
Assume that $V_n^2/(2\pn B_n)\to 1$ as $n\to\infty$.
Then, for all $n\in \N$ large enough, we have
\begin{align}\label{eq:29a}
 &  \|\varphi_n\|^2\ge 2^{\pn+1} \pi B_n^{\pn-1}\pn! \\
\label{eq:29}
&\|\psi_n\|^2\ge \frac{1}{4} \|\varphi_n\|^2\left(1-\frac{1}{\pn!}\int_{B_n
    r_n^2/2}^\infty s^\pn e^{-s} \rd s\right).
\end{align}
\end{lemma}
\begin{proof}
For $\psi_n=(\psi_{n,1},\psi_{n,2})^{\rm T}$ we compute for $n\in\N$
so large that $V_n^2/(2\pn B_n)\le 3$
\begin{align*}
  \|\psi_n\|^2&\le \|\varphi_n\|^2= \|\varphi_{n,1}\|^2+\|\varphi_{n,2}\|^2\\
&= 2B_n^{2\pn}\pi \int_0^\infty s^{2\pn} e^{-B_ns^2/2} s \rd s +
2V_n^2 B_n^{2(\pn-1)}\pi \int_0^\infty s^{2(\pn-1)} e^{-B_ns^2/2} s
\rd s
\\
&= 2^{\pn+1} \pi B_n^{\pn-1} \left(\int_0^\infty s^\pn e^{-s} \rd s+
\tfrac{V_n^2}{2B_n} \int_0^\infty s^{\pn-1} e^{-s} \rd s\right)\\
&= 2^{\pn+1} \pi B_n^{\pn-1} \left(1+\tfrac{V_n^2}{2B_n\pn}
\right) \pn!. 
\end{align*}
Noting that 
\begin{align*}
   2^{\pn+1} \pi B_n^{\pn-1}\pn! \le 2^{\pn+1} \pi B_n^{\pn-1} \left(1+\tfrac{V_n^2}{2B_n\pn}
\right) \pn!\le  2^{\pn+3} \pi B_n^{\pn-1}\pn!,
\end{align*}
we get  an upper and lower bound for $\|\varphi_n\|^2$. In particular,
we obtain \eqref{eq:29a}. Moreover, 
\begin{align*}
  \|\psi_n\|^2\ge\|\psi_{n,1}\|^2 &\ge 2B_n^{2\pn} \pi \int_0^{r_n}
  (s^2)^\pn   e^{-B_ns^2/2} s \rd s\\
  &= 2^{\pn+1}B_n^{\pn-1} \pi \left(\pn!-\int_{B_n r_n^2/2}^\infty s^\pn e^{-s} \rd s\right)\\
  &\ge\frac{1}{4} \|\varphi_n\|^2\left(1-\frac{1}{\pn!}\int_{B_n
    r_n^2/2}^\infty s^\pn e^{-s} \rd s\right).
\end{align*}
This finishes the proof.
\end{proof}
\begin{proof}[Proof of Theorem \ref{thm3a}]
In this proof we use Lemma \ref{landau} for $\pn=k$ $(n\in\N)$, where $k$ is some fixed natural
number. We choose the localization radii to be given by
\begin{equation}
  \label{eq:30}
  r_n=B_n^{(\varepsilon-1)/2}.
\end{equation}
Since $r_n\to 0$ as $n\to\infty$ we can  assume that the $\psi_n$  have
disjoint support, for otherwise we can extract a subsequence
satisfying this property.
In view of Remark \ref{velocity} we see  that, as $n\to\infty$,
\begin{equation}
  \label{eq:31}
  \frac{1}{\pn!}\int_{B_n
    r_n^2/2}^\infty s^\pn e^{-s} \rd
  s=\frac{1}{k!}\int_{B_n^\varepsilon/2}^\infty s^k e^{-s} \rd s\to 0.
\end{equation}
Then, according to Lemma \ref{landau} there exists an $N>0$ such that,
for all $n>N$,
\begin{align}
  \label{eq:32}
  \|\varphi_n\|\le 8 \|\psi_n\|.
\end{align}
Next we estimate the corresponding terms of \eqref{weyl}. A simple
calculation shows that 
\begin{align*}
  \|(D_{{\mathbf
     A}_n}+V_n)\hat{\varphi}_n\|^2\le |(V_n^2-2kB_n)/V_n|
 \|\varphi_{n,2}\|^2\le 8 |(V_n^2-2kB_n)/V_n|
 \|\psi_n\|^2.
\end{align*}
Thus, by  \eqref{con2a},   $\|(D_{{\mathbf
     A}_n}+V_n)\hat{\varphi}_n\|^2/\|\psi_n\|^2$ converges to $0$ as
 $n\to \infty$. 
Moreover,
using \eqref{eq:29a}, we get for $n$ sufficiently large 
\begin{align*}
  \|({\boldsymbol \sigma}\cdot\nabla
\chi_n)\hat{\varphi}_n\|^2&
\le r_n^{-2} \|\nabla\chi\|_\infty
\int_{r_n\le |\bx|\le 2r_n} |\varphi_n(\bx)|^2 \rd \bx
\\
&\le B_n^{1-\varepsilon} \|\nabla\chi\|_\infty 2^{(k+3)} \pi B_n^{k-1} 
\int_{B_n^\epsilon/2}^\infty s^k e^{-s} \rd s\\
&\le 2^5 \|\psi_n\|^2 \|\nabla \chi\|_\infty
B_n^{1-\varepsilon} \frac{1}{k!}\int_{B_n^\epsilon/2}^\infty s^k e^{-s} \rd s.
\end{align*}
Thus, $ \|({\boldsymbol \sigma}\cdot\nabla
\chi_n)\hat{\varphi}_n\|/\|\psi_n\|\to 0$ as $n\to\infty$. For the
last two terms of \eqref{weyl} we recall that $|\nabla B|$ and $|\nabla
V|$ grow with rate $0$ at infinity. Therefore, applying 
the mean value theorem we find a constant $C>0$ such that
\begin{align*}
\|{\boldsymbol \sigma}\cdot({\mathbf
    A}_n-\widetilde{\mathbf A}_n)\chi_n \hat{\varphi}_n\|^2 \le C^2
  |\nabla B(\bx_n)|^2 r_n^4  \|\hat{\varphi}_n\|^2\le 8 C^2
  \left|\frac{\nabla B(\bx_n)}{B_n^{1-\varepsilon}}\right|^2 \|\psi_n\|^2
\end{align*}
and
\begin{align*}
\|(V-V_n)\chi_n\hat{\varphi}_n\|^2\le  C^2
  |\nabla V(\bx_n)|^2 r_n^2  \|\hat{\varphi}_n\|^2
\le 8 C^2 \left|\frac{\nabla V(\bx_n)}{|V_n|^{1-\varepsilon}}\right|^2
\left(\frac{V_n^2}{B_n}
\right)^{1-\varepsilon} \|\psi_n\|^2.
\end{align*}
Therefore, we obtain the desired result in view of \eqref{con1a} and
since $V^2_n/B_n$ is uniformly bounded for large $n$.
\end{proof}
\begin{proof}[Proof of Theorem \ref{thm3}] In this case we apply Lemma
  \ref{landau} for
  $p_n=n$  $(n\in \N)$. We set
 \begin{align}
    \label{eq:28}
    r_n:=\sqrt{2n^{(1+\varepsilon)}/B_n}.
  \end{align} 
Using  that $s^n
e^{-s/2}\le (2n)^n e^{-n}$ we get
\begin{align*}
  \frac{1}{n!}\int_{n^{(1+\varepsilon)}}^\infty s^n e^{-s} \rd s\le
  2\frac{(2n)^n  e^{-n^{(1+\varepsilon)}/2-n}}{n!}\le {\rm
    exp}\big(n\ln(2n) -n^{(1+\varepsilon)}/2-n\big). 
\end{align*}
The last term tends to zero as $n\to \infty.$ Hence, we find
$N>0$  so large that
\begin{align}\label{vera}
  {\rm exp}\big(n\ln(2n) -n^{(1+\varepsilon)}/2-n\big)<1/2,\quad \mbox{for
   all}\quad n>N.
\end{align}
As a consequence we get in view of Lemma \ref{landau}, for all $n>N$,
\begin{align}\label{comon}
  \|\varphi_n\|^2\le 8\|\psi_n\|^2.
\end{align}
We now  estimate the radii.  In view of Remark \ref{growth} we may assume
that, for all $n>N$,
\begin{align}\label{rela}
  \frac{1}{2}\le \frac{V^2_n}{2nB_n}\le 2.
\end{align}
This together with \eqref{con0}
 implies that
\begin{align*}
  r_n^{-2}\le 2^{\kappa-1}\alpha 
  n^{\kappa-1-\varepsilon} 
\le n^{\kappa},
\end{align*}
for $n\in\N$ large enough.
Therefore, using \eqref{eq:29a}, we get 
\begin{align*}
  \|({\boldsymbol \sigma}\cdot\nabla
\chi_n)\hat{\varphi}_n\|^2
&\le n^{\kappa} \|\nabla\chi\|_\infty 2^{(n+3)} \pi B_n^{n-1} 
\int_{n^{1+\varepsilon}}^\infty s^n e^{-s} \rd s\\
&\le  4\|\nabla\chi\|_\infty  \|\varphi_n\|^2 n^{\kappa} {\rm
  exp}\big(n\ln(2n) -n^{(1+\varepsilon)}/2-n\big)\\
& \le  2^{5} \|\nabla\chi\|_\infty  \|\psi_n\|^2   n^{\kappa} {\rm
  exp}\big(n\ln(2n) -n^{(1+\varepsilon)}/2-n\big).
\end{align*}
Thus, $\|({\boldsymbol \sigma}\cdot\nabla
\chi_n)\hat{\varphi}_n\|^2/\|\psi_n\|^2\to 0$ as $n\to\infty$.

For the last two terms in \eqref{weyl} we use \eqref{eq:28} and
\eqref{rela} to get, for $\nu \in[0,1]$,
\begin{align}
  \label{eq:33}
\frac{r_n^2}{|\bx|^{2 \nu}}=\frac{2}{B_n}
\frac{n^{1+\varepsilon}}{|\bx_n|^{2 \nu}}\le 
\frac{2^{2+\varepsilon}}{B_n |\bx_n|^{2 \nu}}   \left( \frac{V_n^2}{2B_n}\right)^{1+\varepsilon}.
\end{align}
This implies by \eqref{con24} that for $\nu\in [0,1]$ the ratio
$r_n/|\bx_n|^{\nu}\to 0$ as $n\to\infty$. Hence, on the one hand, we see that
the support of the $\psi_n$ are mutually disjoint (at least for a
subsequence).  On the other hand, since $|\nabla V|$ and $|\nabla B|$ vary
with rate $\nu \in [0,1]$, this enable us to use the mean value
theorem to obtain that
\begin{align*}
  \|{\boldsymbol \sigma}\cdot({\mathbf
    A}_n-\widetilde{\mathbf A}_n)\chi_n \hat{\varphi}_n\|^2
&\le C^2
  |\nabla B(\bx_n)|^2 r_n^4  \|\hat{\varphi}_n\|^2\\
&\le 2^{7+2\varepsilon} C^2 \left[\frac{|\nabla B(\bx_n)|}{B_n}
  \left(\frac{V_n^2}{2B_n}\right)^{1+\varepsilon}\right]^2
\|{\psi}_n\|^2, 
\end{align*}
where we use again \eqref{eq:28} combined with
\eqref{rela} and \eqref{comon}. Analogously, we get
\begin{align*}
  \|(V-V_n)\chi_n\hat{\varphi}_n\|^2&\le  C^2
  |\nabla V(\bx_n)|^2 r_n^2  \|\hat{\varphi}_n\|^2\\
&\le  2^{6+2\varepsilon} C^2
\left[\frac{|\nabla
    V(\bx_n)|}{V_n}\left(\frac{V^2_n}{2B_n}\right)^{1+\varepsilon}\right]^2
\|\psi_n\|^2.
\end{align*}
Hence, in view of \eqref{con2}  we get that
$
  \|(\DA+V)\psi_n\|/\|\psi_n\|\to 0$ as  $n\to \infty
$ which proves the theorem.
\end{proof}
\begin{appendix}
\section{Results  for rotationally symmetric potentials}
In this appendix we study some properties of the Dirac operator
$H = \DA + V$ when  $V$ and $B$ are rotationally symmetric.  Let  $\bf
A$ be  given by the rotational  gauge
\begin{equation}\label{goischt}
{\bf A}({\bf x}) := \frac{A(r)}{r} 
\begin{pmatrix} -x_2 \\ x_1 \end{pmatrix}\,, \quad
A(r) =  \frac{1}{r} \int_0^r B(s) s \mbox{d}s, 
\end{equation}
where $r=|\bx|$. 
We can decompose $H$
into a direct sum of operators on the $j$-th angular momentum eigenspace,
i.e., there is a unitary map 
$$U: L^2(\R^2,\C^2) \to \bigoplus_{j\in \Z}L^2(\R^+,\C^2; \mbox{d}r)$$
such that 
$ UHU^* =\bigoplus_{j\in \Z} h_j$
with
$$h_j := - \ri \sigma_2 \partial_r  + \sigma_1 \left( A(r) - \frac{m_j}{r}\right) + v(r)
 \quad \  \mbox{on} \ L^2(\R^+,\C^2; \mbox{d}r), $$
where  $v(|\bx|):= V(\bx)$ and $m_j = j +1/2$ , $j \in \Z$  (see
e.g. \cite{Thaller}). Since H is essentially self-adjoint on
$\core$ we deduce that $h_j$  is also essentially self-adjoint on
$\big(U\core\big)_j \subset L^2(\R^+,\C^2; \mbox{d}r)$ for any $j \in \Z$. 
\begin{proposition}\label{purepoint}
Assume that $A \in C^1(\R^+,\R)$ and $v \in C(\R^+,\R)$  
are such that the conditions  
\begin{eqnarray}
 && |A(r)|\to \infty\quad\mbox{as}\quad r \to \infty, \label{pp1}\\
 && \frac{A'(r)}{A^2(r)} 
  \to 0 \quad\mbox{as}\quad r \to \infty, \label{pp2}\\
  && \limsup_{r \to \infty} \left| \frac{v(r)}{A(r)} \right| < 1, 
\end{eqnarray}
are fulfilled. Then, for all $j \in \Z$, the spectrum of $h_j$ is
purely discrete. In
particular, the spectrum of $\DA+V$ is pure point.
\end{proposition}
\begin{proof}
  Since $\limsup_{r \to \infty} v^2(r)/A^2(r) < 1$ we find constants
  $\mu > 1$ and $r_0 >0$ such that $A^2(r) > \mu v^2(r)$ for $r >
  r_0$. We pick $\lambda \in (\mu^{-1},1)$ and define, for $j \in \Z$,
  the potential function $ w_j$ on $\R^+$ as
$$w_j(r) :=  \left( A^2(r) - 2\frac{m_j}{r} A(r) 
  - \lambda^{-1}v^2(r)\right) +
\begin{pmatrix}
  -A'(r) & 0 \\
  0 & A'(r)
\end{pmatrix}.
$$
For any $j \in \Z$ the matrix valued potential $w_j$ is real-valued
and diagonal. Moreover, the diagonal entries of $w_j$ are bounded from
below and converge to $+\infty$ as $r \to \infty$ due to \eqref{pp1}
and \eqref{pp2}. As a consequence we
get, for every $j \in \Z$, a self-adjoint operator
$$s_j :=  
\begin{pmatrix}
  -\partial_r^2 + \frac{j^2-1/4}{r^2} & 0 \\
    0    & -\partial_r^2 + \frac{(j+1)^2-1/4}{r^2}
\end{pmatrix}
 + w_j(r)
 \quad \  \mbox{on} \quad L^2(\R^+,\C^2; \mbox{d}r)\, $$
 with form core $\big(U\core\big)_j$.
Moreover, $s_j$ is bounded from 
below and has purely discrete spectrum.
In order to show that $h_j$ has also purely discrete spectrum we
observe that, for any 
 $\psi \in \big(U\core\big)_j$,
\begin{align*}
 \|h_j\psi\|^2 
&\ge   (1-\lambda) \left\|\left(- \ri \sigma_2 \partial_r  + 
                 \sigma_1 \left( A(r) - \frac{m_j}{r}\right)\right)\psi \right\|^2
     + ( 1- \lambda^{-1})\|v\psi\|^2 \\
&= (1-\lambda)\left[\left\|\left(- \ri \sigma_2 \partial_r  + 
                    \sigma_1 \left( A(r) - \frac{m_j}{r}\right)\right)\psi \right\|^2
     - \lambda^{-1}\|v\psi\|^2 \right] \\[0.2cm]
&= (1-\lambda)\sps{\psi}{s_j \psi}\,,
\end{align*}
which is equivalent to $h_j^2 \ge (1-\lambda) s_j$, for $j \in
\Z$. An application of the min-max principle (see \cite[Section
XIII.1]{ReedSimon1978}) gives that  $h_j^2$ has purely discrete
spectrum. This implies the claim for $h_j$.
\end{proof}
 \section{Some technical tools}
\begin{lemma}\label{schneider}
Assume $V\in C^1(\R^2,\R)$, $B\in C(\R^2,\R)$ and 
${\bf A}\in C^1(\R^2,\R^2)$ with $B= \curla$ on $\R^2$, such that the conditions  
\begin{eqnarray}
&& V(\bx), B(\bx) \to \infty\quad \mbox{as} \quad |\bx| \to \infty \label{schneiderc1}  \\
&& \left| \frac{\nabla V (\bx)}{V(\bx)} \right|
     \to 0 \quad \mbox{as} \quad |\bx| \to \infty \label{schneiderc2}  \\
&& \limsup_{|\bx|\to \infty}
  \frac{V^2(\bx)}{2B(\bx)}<1  \label{schneiderc3} 
\end{eqnarray}
are fulfilled. Then, there is a constant $\eta \in (0 ,1)$ such that
for any $\delta \in(0,1)$ we can find an electric potential $
\hat{V}\in C^1(\R^2,\R)$ and a magnetic field $\hat{B}\in
C(\R^2,\R)$ satisfying conditions \eqref{c1}-\eqref{c3} such that
\begin{center} $\sigma(D_{\bf A} + V )$ is discrete if and only if
$\sigma(D_{\hat{\bf A}} +\hat{V} )$ is discrete.
\end{center}
Here $\hat{\bf
  {A}}\in C^1(\R^2,\R^2)$ is a magnetic vector
potential satisfying $\hat{B} = \mbox{curl}\ \hat{\bf {A}}$ on $\R^2$.
\end{lemma}
\begin{remark}
 An analogue statement holds true for potentials $V$ with $V(\bx) \to -\infty $ as $|\bx| \to \infty$
 and magnetic fields $B$ with $B(\bx) \to -\infty $ as $|\bx| \to \infty$
\end{remark}
\begin{proof}
 We first construct the potential $\hat{V}$.
Without lost of generality we assume  that $V(\bx) \geq 0$ for all
$\bx \in \R^2$ (otherwise start with $ V + c $ for some constant
$c$).
Due to assumption \eqref{schneiderc3} we  find $R_1 \geq 1$ and $\eta \in (0 ,1)$ such that 
\begin{align}\label{joe2}
V^2(\bx) < 2(1 - \eta) B(\bx), \quad|\bx| \geq R_1.
\end{align} 
Similarly by \eqref{schneiderc1} and \eqref{schneiderc2}
we  find for any $\delta\in(0,1)$ a constant $R_2 > R_1$ such that
\begin{align}\label{joe1}
V(\bx) > 2\delta^{-1}\quad \mbox{and} \quad|\nabla V(\bx)| <
\frac{\delta}{4}\,  V(\bx), \quad |\bx| \geq R_2.
\end{align}
Pick $ \xi \in C^{\infty}(\R^2,\ [\tfrac{1}{2},1]) $ such that
$\xi (\bx)=1/2 $ for $|\bx|\le 1$ and $\xi(\bx)=1$ for $|\bx|\ge 2$.
For $r>R_2$ we set $\xi_r(\bx) := \xi \big(\frac{\bx}{r}\big) $, $\bx
\in \R^2 $. Define for $\bx \in \R^2$
\begin{equation*}
 \hat{V}_r(\bx) :=  \xi_r(\bx) V(\bx) + (1 - \xi_r(\bx) )4 \delta^{-1} (1 + M),
\end{equation*}
where $M := \sup \left\{|\nabla V(\bx)| \mid  |\bx| \leq R_2 \right\}
$. Clearly 
\begin{align}\label{joe4}
 \hat{V}_r(\bx) = V(\bx), \quad|\bx|>2r.
\end{align}
Furthermore, we find  that 
\begin{align}
  \label{joe3}
  \hat{V}_r(\bx) \ge \delta^{-1}, \quad  2\hat{V}_r(\bx) \ge
  V(\bx),\qquad \bx\in \R^2.
\end{align}
A simple computation yields, for any $\bx\in \R^2$,
\begin{align*}
  |\nabla \hat{V}_r(\bx) |\le   |\nabla V(\bx)| +
\frac1r\left( 2 \|\nabla \xi \|_{\infty} + 4 \|\nabla \xi \|_{\infty} \delta^{-1}(1 + M) \right)  \hat{V}_r(\bx).
\end{align*}
Note that  $|\nabla V(\bx)|\le {\rm max}\{M, \delta V(\bx)/4\}, \bx \in \R$.
Hence, using \eqref{joe1} and the definition of $\hat{V}_r$ we get that
$|\nabla V(\bx)|\le \delta \hat{V}_r(\bx) /2$, for all $\bx\in
\R^2$. Thus, we find a constant $r_0>R_2$ so large that
\begin{align}
  \label{eq:222}
   |\nabla \hat{V}_{r_0}(\bx) |\le  \delta \hat{V}_{r_0}(\bx), \quad \bx \in \R^2.
\end{align}
Define  $\hat{V}:=\hat{V}_{r_0}$ and note that $\hat{V}$  fulfills the desired
properties in view of \eqref{joe3} and \eqref{eq:222}.

Next we define the magnetic field as 
\begin{equation*}
\hat{B}(\bx) := 
2 \left( \xi_{2r_0}(\bx) - \tfrac{1}{2} \right) B(\bx)  + \left( 1 - \xi_{2r_0}(\bx) \right) \frac{\hat{V}^2(\bx)}{(1-\eta)} , 
\quad \bx \in \R^2.
\end{equation*}
In view of \eqref{joe4} and  \eqref{joe2} we have 
\begin{align*}
  \hat{B}(\bx)\ge \frac{\hat{V}^2(\bx)}{2 (1-\eta)},\quad \bx\in \R^2.
\end{align*}
Hence, $\hat{V}$ and $\hat{B}$ satisfy conditions
\eqref{c1}-\eqref{c3}.
 
Since the function $B - \hat{B}$ has compact support in $\R^2$, we find a function 
${\bf G}\in C^1(\R^2,\R^2)$ such that $|| \bf{G} ||_{\infty}  < \infty $ and 
$B - \hat{B} =\mbox{curl} \ {\bf G}$ on $\R^2$. We define 
$\hat{\bf A} (\bx) := \bf A ( \bx) - \bf G (\bx)$ for  $\bx \in
\R^2$. By construction 
we know that $(\DA +V) - (D_{\hat{\bf A}} +\hat{V} )$ is bounded
on $\hilbert$. Hence, the resolvent difference of $\DA +V$
and $D_{\hat{\bf A}} +\hat{V}$ is compact if one of the resolvents is
itself 
compact.
From this follows the claim.
\end{proof}
\end{appendix}

\end{document}